\documentclass[journal,onecolumn]{IEEEtran}
\usepackage{setspace}
\usepackage{amsmath,amssymb}
\usepackage{graphicx}
\usepackage{booktabs}
\usepackage{cite}
\usepackage{url}
\usepackage{balance}

\newtheorem{proposition}{Proposition}
\newcommand{\E}{\mathbb{E}}
\newcommand{\Prb}{\mathbb{P}}
\newcommand{\pos}[1]{\left[#1\right]^+}

\graphicspath{{figures/}}

\title{Contact-Limited Throughput of a Buoyless Acoustic-to-LEO Gateway With Anticipatory Preparation}

\author{Muhammad Khalil%
\thanks{Muhammad Khalil is with the School of Engineering, RMIT University, Melbourne, VIC, Australia (e-mail: muhammad.khalil@rmit.edu.au).}}

\begin{document}
\maketitle

\begin{abstract}
A buoyless acoustic-to-LEO gateway with predictable contacts is analyzed. Advanceable cross-medium preparation competes with acoustic collection, while residual RF acquisition consumes contact time. Exact fluid and packet service, a fluid-optimal fixed preparation lead, tandem-queue stability, and one-contact reliability are derived. An explicit service-bias identity shows why treating acquisition as advanceable can select an insufficient lead. For lognormal preparation, conditions are characterized under which variability improves mean service while degrading reliability. Monte Carlo and packet-queue simulations corroborate the analysis with synthetic parameters; in a diagnostic case, accounting for residual RF acquisition raises the actual sustainable rate from 8.776 to 9.079 kbit/s.

\end{abstract}


\section{Introduction}
\IEEEPARstart{U}{nderwater} acoustic networks face limited bandwidth,
propagation delay, and channel variability \cite{stojanovic}. Surface
gateways connect them to above-water infrastructure \cite{ibrahim}, while
hybrid aerial--underwater vehicles (HAUVs) provide a mobile alternative
with demonstrated cross-medium operation \cite{jin,zhen}. Related paths
include BeiDou--acoustic experiments \cite{zheng}, mobile sea--air
gateways \cite{bao}, and autonomous ARGOS pop-up buoys
\cite{carandell}. Relay--aided satellite-acoustic systems have also been
analyzed through outage probability and average bit-error rate
\cite{yang2024}.

Pass anticipation is established: Arvor floats use ephemerides to synchronize surfacing/transmission with predicted Argos passes \cite{andre}. Finite service visits, setup/vacation queues, asynchronous satellite opportunities \cite{alfa,niu,carrhajek}, and contact-plan routing \cite{araniti} are also established. These results provide the basis for the contact-limited service model considered here.

The main design question is how much acoustic collection time should be allocated to pre-contact preparation when only part of the readiness process can precede RF access. A fixed preparation lead is used to separate advanceable preparation from residual RF acquisition. An exact bias identity then quantifies the error caused by treating both components as advanceable. The resulting service law is used to obtain the fluid-optimal lead, the reliability-constrained lead, exact whole-packet service, and the corresponding tandem-queue stability condition. For lognormal preparation time, analytical conditions identify when increased variability improves mean service while reducing reliability. Numerical comparisons evaluate both designs under the same acquisition-constrained service law. The lead uses a supplied contact plan and readiness statistics; trajectory and dynamic surfacing control are outside the scope.

\section{Alternating Gateway With Anticipatory Preparation}
\subsection{Timing and Service Model}
Consider aggregate underwater sources and one HAUV, each with an infinite buffer. A nominal cycle contains a pre-contact interval $U$, an allocated LEO visibility interval $w$, and a recovery/re-entry reservation $Y$:
\begin{equation}
T=U+w+Y. \label{eq:T}
\end{equation}
Successive windows can belong to different satellites, and $T$ denotes the externally supplied service-cycle period. Fig.~\ref{fig:system} summarizes the cross-medium architecture and timing convention.

\begin{figure*}[t]
\centering
\includegraphics[width=0.98\textwidth]{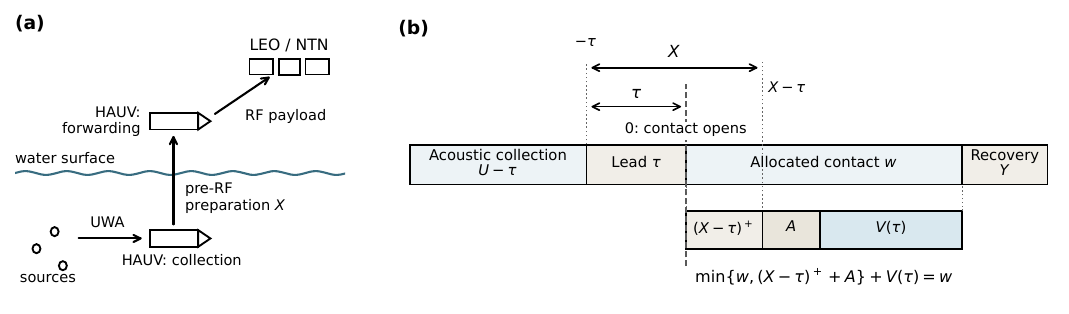}
\caption{System and timing model. (a) The same HAUV collects underwater acoustic (UWA) data and forwards it over LEO/NTN after transition. (b) Preparation starts at $-\tau$; the example has $X>\tau$ and $R(\tau)<w$. Acquisition begins after both contact opening and preparation completion. Unfinished readiness is truncated at closure.}
\label{fig:system}
\end{figure*}

Let $\tau\in[0,U]$ be a fixed pre-contact preparation lead.
Acoustic reception occupies $U-\tau$ seconds and then stops as advanceable preparation begins $\tau$ seconds before the contact opening.
Let $X_k\ge0$ denote advanceable pre-RF preparation that is incompatible with continued acoustic collection, such as water exit and platform/antenna readiness, and let $A_k\ge0$ denote the residual over-the-air acquisition/access time. We assume the same minimum-elevation mask for acquisition and payload, with no acquisition occurring below that mask. Thus, $A_k$ starts only after both contact opening and the completion of $X_k$; early-completed readiness is maintained until opening. The in-contact readiness overhead is
\begin{equation}
R_k(\tau)=A_k+(X_k-\tau)^+,
\end{equation}
and the usable satellite interval is
\begin{equation}
\boxed{V_k(\tau)=\pos{w-R_k(\tau)}.}
\label{eq:Vtau}
\end{equation}
Lead can absorb $X_k$ but cannot remove $A_k$. We take $(X_k,A_k)$ as exogenous to $\tau$, allowing for dependence within each pair. Attempts unfinished at closure are abandoned, and safe abort/re-entry must fit within $Y$; otherwise, the fixed-cycle model does not apply. A start-time-dependent acquisition $A_k(\tau)$ requires a revised service law; Proposition~1 still applies if its induced mean service is continuous and strictly increasing.

Let $C_U,C_N>0$ be constant effective payload rates during active service, with protocol costs absorbed into these rates. The fluid service opportunities are
\begin{equation}
S_U(\tau)=C_U(U-\tau),\qquad
S_{N,k}(\tau)=C_NV_k(\tau).
\label{eq:fluidservice}
\end{equation}
Define $H(\tau)=\E[V(\tau)]$. For lead-independent $(X,A)$, $V(\tau)$ is sample-wise nondecreasing and 1-Lipschitz, so $H$ is nondecreasing and continuous for any readiness distribution.

For compact closed-form sensitivity and numerical evaluation, we later use a fixed visibility-dependent acquisition reservation $A=a<w$ and a random $X$. With $\bar w=w-a$ and
\begin{equation}
G_X(x)=\E[(x-X)^+]=\int_0^x F_X(t)\,dt, \qquad x\ge0,
\label{eq:G}
\end{equation}
\eqref{eq:Vtau} gives
\begin{equation}
H(\tau)=G_X(\tau+\bar w)-G_X(\tau),
\qquad
p_0(\tau)=1-F_X(\tau+\bar w),
\label{eq:Hdet}
\end{equation}
for continuous $X$, where $p_0=\Prb\{V=0\}$. The reservation $a$ is an illustrative acquisition parameter. Random $A$ is covered by \eqref{eq:Vtau} using joint samples with $X$.

\subsection{Fluid-Optimal Preparation Lead}
Under the exogenous i.i.d. cycle model used below, the sustainable fluid-rate supremum is:
\begin{equation}
C_f(\tau)=\frac{1}{T}\min\{C_U(U-\tau),\,C_NH(\tau)\}.
\label{eq:Cftau}
\end{equation}

\begin{proposition}
Assume $H(\tau)$ is strictly increasing on $[0,U]$ and $H(U)>0$. If $C_UU\le C_NH(0)$, then $\tau_f^\star=0$ uniquely maximizes \eqref{eq:Cftau}. Otherwise, there is a unique $\tau_f^\star\in(0,U)$ satisfying
\begin{equation}
\boxed{C_U(U-\tau_f^\star)=C_NH(\tau_f^\star),}
\label{eq:taustar}
\end{equation}
and it uniquely maximizes $C_f(\tau)$.
\end{proposition}

\begin{IEEEproof}
The acoustic term strictly decreases and the satellite term strictly increases. If the acoustic term is no larger at zero, positive lead lowers the minimum. Otherwise their difference changes sign between $0$ and $U$; continuity and strict monotonicity give one crossing. Thus, the minimum increases up to the crossing
and decreases thereafter.
\end{IEEEproof}

For $A=a$ and continuous $X$, $H'(\tau)=F_X(\tau+\bar w)-F_X(\tau)>0$ whenever the CDF has positive mass over each interval $(\tau,\tau+\bar w]$, the proposition applies to the lognormal case below.

\subsection{Bias From Advancing RF Acquisition}
Consider the diagnostic approximation
\begin{equation}
\widetilde V(\tau)=\pos{w-(X+a-\tau)^+},\qquad
\widetilde H(\tau)=\E[\widetilde V(\tau)],
\label{eq:approx}
\end{equation}
which treats the entire readiness time as advanceable before RF access. Since
$\widetilde H(\tau)=G_X(\tau+w-a)-G_X((\tau-a)^+)$,
\begin{equation}
\boxed{\widetilde H(\tau)-H(\tau)
=\int_{(\tau-a)^+}^{\tau}F_X(t)\,dt
\le aF_X(\tau).}
\label{eq:bias}
\end{equation}
The bias is nonnegative and vanishes when the preparation cannot be completed before opening. Let $\widetilde\tau_f^\star$ maximize \eqref{eq:Cftau} with $H$ replaced by $\widetilde H$. Under Proposition~1, the crossing argument provides $\widetilde\tau_f^\star\le\tau_f^\star$: the approximation can select insufficient leads. Both lead selections are therefore evaluated using the actual service rate $C_f$ in \eqref{eq:Cftau}.

\section{Packet Service, Stability, and Reliability}
With fixed serialization times $L/C_U$ and $L/C_N$, the whole-packet service is as follows:
\begin{equation}
K_U(\tau)=\left\lfloor\frac{C_U(U-\tau)}{L}\right\rfloor,
\quad
K_{N,k}(\tau)=\left\lfloor\frac{C_NV_k(\tau)}{L}\right\rfloor.
\label{eq:K}
\end{equation}
Let $F_{R_\tau}$ denote the CDF of $R(\tau)$. The tail-sum identity gives the exact satellite mean
\begin{equation}
\boxed{
\E[K_N(\tau)]
=\sum_{j=1}^{\lfloor C_Nw/L\rfloor}
F_{R_\tau}\!\left(w-\frac{jL}{C_N}\right).
}
\label{eq:tailtau}
\end{equation}
Sample-wise packetization also yields
\begin{equation}
0\le C_NH(\tau)-L\E[K_N(\tau)]<L.
\label{eq:gap}
\end{equation}
For $A=a$, \eqref{eq:tailtau} reduces to a sum of $F_X(\tau+\bar w-jL/C_N)$ over $j\le\lfloor C_N\bar w/L\rfloor$.

Let $N_k$ packets arrive at the source queue at the cycle start, with $\E[N_k]=\lambda T/L$. With acoustic transfer before the satellite contact,
\begin{align}
B_k&=\min\{Q_{U,k}+N_k,K_U(\tau)\},\nonumber\\
Q_{U,k+1}&=Q_{U,k}+N_k-B_k,\label{eq:queues}\\
Q_{H,k+1}&=(Q_{H,k}+B_k-K_{N,k}(\tau))^+.\nonumber
\end{align}

\begin{proposition}
If $\{N_k\}$ and $\{(X_k,A_k)\}$ are mutually independent i.i.d. sequences with finite first moments, and service is exogenous, lossless, and work-conserving, both queues admit a proper stationary backlog regime whenever
\begin{equation}
\lambda<C_{\rm pkt}(\tau),\qquad
C_{\rm pkt}(\tau)=\frac{L}{T}\min\{K_U(\tau),\E[K_N(\tau)]\}.
\label{eq:Cpkt}
\end{equation}
If $\lambda>C_{\rm pkt}(\tau)$, both queues cannot be stable. No general assertion is made at equality.
\end{proposition}

\begin{IEEEproof}
Strict $\E[N_k]<K_U(\tau)$ stabilizes the first queue; its stationary departures are ergodic with $\E[B_k]=\E[N_k]$. Independence from exogenous readiness makes $B_k-K_{N,k}$ stationary and ergodic. Loynes' theorem \cite{loynes} stabilizes the second recursion when its mean is negative. Violation of either mean-service bound precludes joint stability.
\end{IEEEproof}

Equation \eqref{eq:Cpkt} specializes standard stability theory to the cross-medium service law. Packet-optimal and fluid-optimal leads can differ: $K_U(\tau)$ changes in integer steps, whereas $\E[K_N(\tau)]$ is generally continuous for continuous $X$. Stability here does not assert a finite mean backlog.

For a positive backlog $b$ already at the HAUV when preparation begins, $0<b\le C_Nw$,
\begin{equation}
p_c(b;\tau)=F_{R_\tau}\!\left(w-\frac{b}{C_N}\right),
\label{eq:pc}
\end{equation}
and $p_c=0$ for $b>C_Nw$. Thus, with the quantile definition $q_p(R_\tau)=\inf\{r:F_{R_\tau}(r)\ge p\}$, the largest fluid backlog meeting target $1-\varepsilon$ is
\begin{equation}
\boxed{b_\varepsilon(\tau)=C_N\pos{w-q_{1-\varepsilon}(R_\tau)}.}
\label{eq:beps}
\end{equation}
This metric gives the payload that can be completed within one satellite contact at the specified reliability. A zero value means that no positive backlog meets the target. For $A=a$ and continuous $X$,
\begin{equation}
b_\varepsilon(\tau)=C_N\pos{\bar w-(q_{1-\varepsilon}(X)-\tau)^+}.
\label{eq:bepsdet}
\end{equation}
For $0<b\le C_N\bar w$, the minimum lead that satisfies $p_c(b;\tau)\ge1-\varepsilon$ and the corresponding reliability-constrained fluid optimum are
\begin{equation}
\begin{aligned}
\tau_\varepsilon(b)&=\pos{q_{1-\varepsilon}(X)+\frac{b}{C_N}-\bar w},\\
\tau_{f,\varepsilon}^\star(b)&=\max\{\tau_f^\star,\tau_\varepsilon(b)\}.
\end{aligned}
\label{eq:taureliability}
\end{equation}
Equation~\eqref{eq:taureliability} applies when $\tau_\varepsilon(b)\le U$; otherwise, the requested one-contact reliability is infeasible under the fixed cycle. The second equality follows because Proposition~1 makes $C_f(\tau)$ increase up to $\tau_f^\star$ and decrease thereafter.

\subsection{Readiness Variability With In-Contact Acquisition}
For fixed $a$, let $X\sim\mathrm{LN}(\ln m_X-\sigma^2/2,\sigma^2)$ have a fixed mean $m_X$ and $\sigma=\sqrt{\ln(1+c_X^2)}>0$, where $c_X$ is the coefficient of variation. For $x>0$, $G_X(x)=x\Phi(z_x)-m_X\Phi(z_x-\sigma)$, where $z_x=(\ln x-\ln m_X+\sigma^2/2)/\sigma$; $\Phi$ and $\phi$ denote the standard normal CDF and density.

\begin{proposition}
For $\tau>0$,
\begin{equation}
\frac{\partial H(\tau)}{\partial\sigma}>0
\iff
m_X>e^{-\sigma^2/2}\sqrt{\tau(\tau+\bar w)},
\label{eq:varsmean}
\end{equation}
while at $\tau=0$ the derivative is strictly positive. Moreover,
\begin{equation}
\frac{\partial p_0(\tau)}{\partial\sigma}>0
\iff
m_X<(\tau+\bar w)e^{-\sigma^2/2}.
\label{eq:varsp0}
\end{equation}
For $0<\varepsilon<1/2$ and $0<b_\varepsilon(\tau)<C_N\bar w$,
\begin{equation}
\frac{\partial b_\varepsilon(\tau)}{\partial\sigma}<0
\iff
\sigma<\Phi^{-1}(1-\varepsilon).
\label{eq:varsbudget}
\end{equation}
Hence a nonempty region exists in which variability raises mean satellite service while worsening both missed-contact probability and reliable payload.
\end{proposition}

\begin{IEEEproof}
At fixed $m_X$, $\partial G_X(x)/\partial\sigma=x\phi(z_x)$. Therefore
$\partial H/\partial\sigma=(\tau+\bar w)\phi(z_{\tau+\bar w})-\tau\phi(z_\tau)$, whose Gaussian-density ratio gives \eqref{eq:varsmean}. Differentiating $p_0=1-\Phi(z_{\tau+\bar w})$ gives \eqref{eq:varsp0}. Finally, $q_{1-\varepsilon}(X)=m_X\exp[-\sigma^2/2+\sigma\Phi^{-1}(1-\varepsilon)]$; differentiating the interior branch of \eqref{eq:bepsdet} gives \eqref{eq:varsbudget}. Saturation boundaries are treated one-sidedly.
\end{IEEEproof}

\section{Numerical Evaluation}
For a circular overhead pass over a spherical Earth, neglecting Earth rotation,
\begin{equation}
\begin{aligned}
\psi_{\max}&=\cos^{-1}\!\left[\frac{R_E}{R_E+h}\cos\theta_{\min}\right]-\theta_{\min},\\
n&=\sqrt{\frac{\mu_E}{(R_E+h)^3}},\qquad w=\frac{2\psi_{\max}}{n}.
\end{aligned}
\label{eq:contact}
\end{equation}
Using $R_E=6378.137$ km, $\mu_E=398600.4418$ km$^3$/s$^2$, $h=600$ km, and $\theta_{\min}=30^\circ$ results in $w=247.141$ s. This value represents the geometric visibility duration for the assumed overhead pass and elevation mask. The service abstraction does not model the detailed NTN RF performance requirements in \cite{3gpp381015}.

Unless stated otherwise, $U=600$ s, $Y=15$ s, $T=862.141$ s, $C_U=20$ kbit/s, $C_N=100$ kbit/s, $L=8000$ bits, and $c_X=0.25$. To keep RF acquisition distinct from advanceable preparation, the reference uses the illustrative reservation $a=10$ s, giving $\bar w=237.141$ s, and sets $m_X=\chi\bar w$. $\bar w=237.141$ s, and sets $m_X=\chi\bar w$. The service parameters, including $a$, are synthetic and used for numerical evaluation.

\begin{figure*}[t]
\centering
\includegraphics[width=0.98\textwidth]{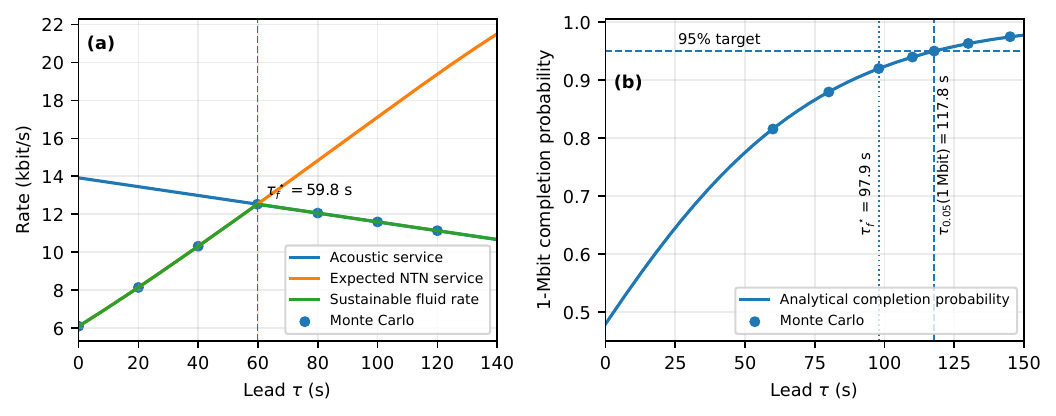}
\caption{Positive-lead performance and reliability. (a) Service tradeoff for $\chi=0.8$, $a=10$ s, and $c_X=0.25$; the analytical optimum is $\tau_f^\star=59.819$ s. (b) One-contact completion probability for a 1-Mbit HAUV backlog at $\chi=1$; the 95\% reliability lead $\tau_{0.05}(1\,\mathrm{Mbit})=117.788$ s exceeds the fluid optimum $97.909$ s. Markers use $5\times10^5$ independent lognormal samples per point.}
\label{fig:performance}
\end{figure*}

Table~\ref{tab:lead} reports the analytical fluid-optimal leads, while Fig.~\ref{fig:performance}(a) shows the full $\tau$ tradeoff at $\chi=0.8$. The lead increases the sustainable fluid rate from $6.100$ to $12.531$ kbit/s while reducing $p_0$ from $15.165\%$ to $2.601\%$. A $5\times10^5$-sample PCG64 Monte Carlo check at $\tau_f^\star$ gives $\E[V]=108.156\pm0.125$ s (95\% interval), versus the analytical $108.036$ s. The 95\%-reliable payload becomes $2.102$ Mbit. At $\chi=1$, the fluid-optimal lead still leaves $p_0=6.340\%$; indeed, $\tau_{0.05}(0^+)=107.788$ s exceeds $\tau_f^\star=97.909$ s. Fig.~\ref{fig:performance}(b) further shows that a 1-Mbit backlog reaches 95\% completion only at $\tau_{0.05}(1\,\mathrm{Mbit})=117.788$ s, so the reliability-constrained lead can be strictly larger than the mean-throughput optimum.

\begin{table}[t]
\caption{Acquisition-Aware Fluid-Optimal Lead ($a=10$ s)}
\label{tab:lead}
\centering
\scriptsize
\begin{tabular}{@{}ccccc@{}}
\toprule
$\chi$ & $\tau_f^\star$ [s] & $C_f(0)/C_f(\tau_f^\star)$ & $p_0(0)/p_0(\tau_f^\star)$ & $b_{0.05}(\tau_f^\star)$\\
& & [kbit/s] & [\%] & [Mbit]\\
\midrule
0.7 & 40.422 & 8.437/12.981 & 5.801/1.352 & 3.611\\
0.8 & 59.819 & 6.100/12.531 & 15.165/2.601 & 2.102\\
1.0 & 97.909 & 2.695/11.648 & 45.101/6.340 & 0\\
\bottomrule
\end{tabular}
\end{table}

Keeping $m_X=189.713$ s fixed, increasing $a$ from $0$ to $10$ and then to $20$ s changes $\tau_f^\star$ from $51.522$ to $59.819$ and then to $68.115$ s, and $C_f(\tau_f^\star)$ from $12.724$ to $12.531$ and then to $12.339$ kbit/s. At $\tau=59.819$ s, $F_X(\tau)=2.50\times10^{-6}$, early completion is rare, and \eqref{eq:bias} predicts little separation from the all-advanceable approximation at this reference optimum.

Fig.~\ref{fig:validation}(a) uses $C_N=40$ kbit/s, keeps the same $X$ distribution, and varies $a$ from $0$ to $40$ s. At $a=40$ s, $\widetilde\tau_f^\star=191.009$ s, and $\tau_f^\star=208.648$ s, evaluated under the actual service law, their rates are $8.776$ and $9.079$ kbit/s, respectively, indicating a $3.45\%$ gain. The corresponding packet thresholds are $8.7704$ and $9.0726$ kbit/s, so an $8.9$-kbit/s load is sustainable only under the acquisition-aware design. This comparison isolates the effect of modeling residual RF acquisition separately from advanceable preparation.

At the $\chi=0.8$, $C_N=100$-kbit/s fluid design point, $C_{\rm pkt}(\tau_f^\star)=12.5266$ kbit/s, within $L/T=0.00928$ kbit/s of $C_f$. Fig.~\ref{fig:validation}(b) uses $\tau=40$ s and $\lambda=11$ kbit/s: the source and satellite packet rates are $12.9909$ and $10.3108$ kbit/s, predicting HAUV drift of $0.68920$ kbit/s. Across 20 independent $200{,}000$-cycle runs, excluding the first $20{,}000$ cycles from drift estimation, the measured drift is $0.69054\pm0.00406$ kbit/s (95\% Student-$t$ half-width). For the same load, $\tau=59.819$ s places both stages below capacity; Fig.~\ref{fig:validation}(c) shows representative running mean backlogs after the same warm-up. The run-averaged HAUV backlog is $5.752\pm0.049$ Mbit (95\% interval across 20 runs).

\begin{figure*}[t]
\centering
\includegraphics[width=0.98\textwidth]{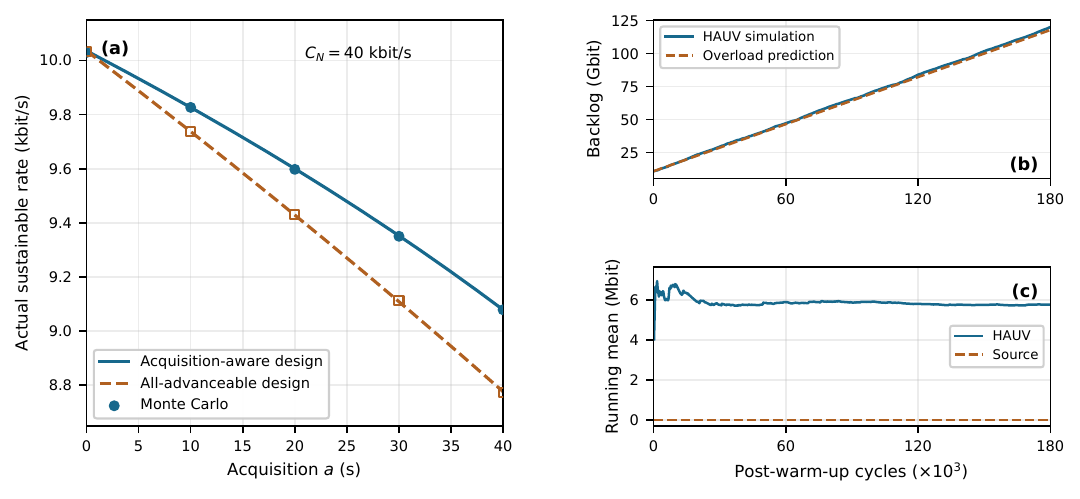}
\caption{Design and queue validation with $m_X=189.713$ s and $c_X=0.25$. (a) Actual rates under both lead designs, $C_N=40$ kbit/s; markers use $5\times10^5$ samples per design/point. (b) Overloaded HAUV queue and predicted drift, $C_N=100$ kbit/s, $a=10$ s, $\tau=40$ s, $\lambda=11$ kbit/s. (c) Running mean HAUV/source backlogs at $\tau=59.819$ s with the same load and rates, after warm-up.}
\label{fig:validation}
\end{figure*}

Queue arrivals are Poisson with a mean of $\lambda T/L$.
NumPy's \texttt{Generator} with \texttt{PCG64} uses
\texttt{SeedSequence}(20260905) as the root, with spawned
substreams for independent Monte Carlo points and queue runs.
The simulations agree with the analytical service predictions.
An ideal persistent surface gateway, with simultaneous acoustic
reception and negligible acquisition, has
$C_B=\min\{C_U,C_Nw/T\}=20$ kbit/s at $C_N=100$ kbit/s.
This value serves as an upper benchmark because the HAUV remains subject to the alternating acoustic duty cycle.

\section{Conclusion}
The analysis separates advanceable preparation from residual RF acquisition in the fixed-lead service model. The bias identity shows when treating all readiness as advanceable underestimates the required lead. The numerical results show the corresponding throughput loss under the actual service law. Reliability can require more lead than mean throughput, and lognormal variability can improve mean service while degrading reliability. Practical deployment requires measured readiness and goodput parameters, feasible contact and recovery schedules, and consideration of elevation-dependent service, Doppler acquisition, finite buffers, and vehicle endurance. The numerical results are conditional on the service model and parameters used here.

\end{document}